\documentclass[10pt, reqno]{amsart}
\usepackage{amssymb}
\usepackage{a4wide}
\usepackage{graphics}
\newtheorem{definition}{Definition}[section]
\newtheorem{theorem}[definition]{Theorem}
\newtheorem{proposition}[definition]{Proposition}
\newtheorem{corollary}[definition]{Corollary}

\newtheorem{remark}[definition]{Remark}

\numberwithin{equation}{section}
\newcommand \Hbf {\mathcal H}
\newcommand \nablabf {\boldsymbol \nabla}
\newcommand \ra \rangle
\newcommand \la \langle
\newcommand \dotn {\dot n}
\newcommand \del        \partial
\newcommand \be   {\begin{equation}}
\newcommand \ee   {\end{equation}}
\newcommand \bfF {\mathbf F}
\def\M{\mathcal{M}}

\def\Q{\mathcal{Q}}

\def\L{\mathcal{L}}
\def\RR     {\mathbb R}
\def\Mcal  {{\mathcal M}^{3+1}}
\def\bg{\mathbf{g}}
\def\bR{\mathbf{R}}
\def\bW{\mathbf{W}}
\def\bN{\mathbf{N}}
\def\bD{\mathbf{D}}

\def\bRic{\mathbf{Ric}}
\def\bm{\mathbf{m}}
\def\bu{\mathbf{u}}
\def\bk{\mathbf{k}}
\def\bB{\mathbf{B}}
\def\bE{\mathbf{E}}
\def\bH{\mathbf{H}}
\def\bP{\mathbf{P}}
\def\bfe{\mathbf{e}}

\newcommand\g{\ensuremath{\mathbf{g}}}
\newcommand\T{\ensuremath{\mathbf{T}}}
\newcommand\gN{\ensuremath{\mathbf{g}_{\bN}}}
\newcommand{\bel}[1]{\be\label{#1}}
\def\a{\alpha}
\def\b{\beta}
\def\d{\delta}
\def\e{\epsilon}
\def\eps{\varepsilon}
\def\l{\lambda}
\def\m{\mu}
\def\tr{\operatorname{tr}}
\def\id{\operatorname{id}}
\def\diag{\operatorname{diag}}

\newcommand\bfpi{\ensuremath{\boldsymbol{\pi}}}
\DeclareMathOperator{\grad}{grad}
\begin{document}

\title[New energy inequalities on spacetimes that satisfy a one-sided bound]{
New energy inequalities for tensorial wave equations 
\\
on spacetimes that satisfy a one-sided bound}
%                     Spelling:   "LeFloch"   or   "LeFLOCH"
\author[A.Y. Burtscher, J.D.E. Grant, and P.G. L{\tiny e}Floch]
{Annegret Y. Burtscher$^{1,3}$, \,  James D.E. Grant$^2$,
\and
Philippe G. L{\smaller e}Floch$^3$
}

\date{}

\thanks{\noindent{}$^1$  Fakult\"{a}t f\"{u}r Mathematik, Universit\"{a}t Wien, Nordbergstra\ss e 15, 1090 Wien, Austria.
%E-mail: {\tt Annegret.Burtscher@univie.ac.at.}
\\
Web: {\tt annegretburtscher.wordpress.com}
\\
$^2$ Gravitationsphysik, Fakult\"{a}t f\"{u}r Physik, W\"{a}hringer Stra{\ss}e 17,
Universit\"{a}t Wien, 1090 Wien, Austria.
%E-mail: {\tt James.Grant@univie.ac.at.}
\\
Web: {\tt jdegrant.wordpress.com}
\\
$^3$ Laboratoire Jacques-Louis Lions, Centre National de la Recherche Scientifique,
Universit\'e Pierre et Marie Curie (Paris 6), 4 Place Jussieu, 75252 Paris,
France. 
Blog: {\tt philippelefloch.org}
\\
 Completed in April 2011. To appear in {\tt Communications in Partial Differential Equations.} 
%\newline Revised in September 2011.
}

\begin{abstract}
We consider several tensorial wave equations, specifically the equations of
Max\-well, Yang--Mills, and Weyl fields, posed on a curved spacetime,
and we establish new energy inequalities under certain {\sl one-sided geometric conditions.}
Our conditions restrict the lapse function and deformation tensor of the spacetime foliation,
and turn out to be a {\sl one-sided and integral} generalization of conditions recently proposed by
Klainerman and Rodnianski as providing
a continuation criterion for Einstein's field equations of general relativity.
As we observe it here for the first time, one-sided conditions are sufficient to derive energy inequalities
for certain tensorial equations, provided one takes advantage of some algebraic properties enjoyed by the
natural energy functionals associated with the equations under consideration.
Our method especially applies to the Bel--Robinson energy for Weyl fields, and
our inequalities control the {\sl growth} of the energy in a uniform way, with implied constants depending
on the one-sided geometric bounds, only.
\end{abstract}
\maketitle
% \thispagestyle{empty}

%\vfill

%\tableofcontents

%\newpage

%===================================================================

\section{Introduction}

We consider several tensorial wave equations and, specifically, the equations satisfied by
Maxwell, Yang--Mills and Weyl fields, posed on a curved spacetime.
Given a foliation of this spacetime by spacelike hypersurfaces,
we derive here new energy inequalities based on the observation that
certain one-sided bounds on the geometry of this foliation are sufficient to control the energy functionals
associated with the tensorial equations under consideration.

Our general purpose here is to derive estimates that depend less on norms of geometric quantities
(like the curvature) and more directly on properties such as, for example, the
eigenvalue spectrum of these geometric quantities. In differential
geometry, one often finds that, for instance, manifolds with an upper curvature
bound behave in a different way from manifolds with a lower curvature bound.
Similarly, submanifolds of Riemannian manifolds with positive mean
curvature are quite different from those with negative mean curvature. From a
geometrical point of view, therefore, much information on the manifold is lost by simply
considering the norm of a curvature quantity.

In the present paper, we identify certain minimal geometric conditions, specifically concerning
the lapse function and deformation tensor (defined below) of the spacetime foliation;
interestingly, these conditions turn out to be a {\sl one-sided and integral} generalization of a condition
proposed recently by Klainerman and Rodnianski as a continuation criterion for the Einstein equations
of general relativity~\cite{KR} (see also~\cite{Wang}).
Under these conditions and by taking advantage of
certain algebraic properties satisfied by the natural energy functionals
associated with the tensorial wave equations under consideration,
we are able to establish new energy inequalities leading to a uniform
control of the {\sl growth\/} of the energy,
with implied constants depending on one-sided bounds of the lapse and deformation tensor, only.
Our method especially applies to the Bel--Robinson energy associated with Weyl fields,
which arises in the analysis of Einstein's field equations.
We can thus identify the optimal geometric conditions that
are necessary and sufficient to derive energy inequalities for Weyl fields on curved spacetimes.
It should be recalled here that this tensor plays a central role in the mathematical analysis of the Einstein equations as illustrated by works of
Andersson \cite{Andersson},
Andersson and Moncrief \cite{AM},
Christodoulou and Klainerman \cite{CK},
Klainerman and Nicol\' o \cite{KN},
Klainerman and Rodnianski \cite{KR}, and Reiris \cite{Reiris}.

In addition, in this paper, we demonstrate that our methods may be applied to other classes of tensorial wave equations
and, specifically, we establish new energy inequalities for Maxwell fields and Yang--Mills fields with compact gauge group.
In contrast and somewhat surprisingly, we find that no such result is available for scalar fields satisfying the
(minimally or conformally coupled) wave equation. Our results have some implications in general relativity, as we discuss at the end of this paper.

Throughout this paper, we use the following notation.
The manifold $(\Mcal, \g)$ is a four-dimensional, globally hyperbolic spacetime,
foliated by hypersurfaces $\Hbf_t$ labelled by a time parameter $t \in I$,
where $I$ is an interval of the real line $\RR$.
We denote by $\bD$ the spacetime Levi-Civita connection associated with the metric $\g$, and by
$\bR$ its Riemann curvature tensor.
For definiteness, we assume that the hypersurfaces of the foliation are compact or asymptotically
flat. We denote by $\bN$ the future directed, time-like
unit normal to each slice $\Hbf_t$, and by $\g_t$ the induced
Riemannian metric on $\Hbf_t$ (i.e.~the first fundamental
form). As this is convenient, we do not
distinguish, notationally, between the $(0, 2)$ and $(2, 0)$ form of the
Lorentzian metric, denoting both by $\g$. Using this convention, the {\sl lapse} and {\sl second fundamental form} of the foliation are defined by
\be
\label{nk}
n:=(-\g( dt, dt))^{-1/2}, \qquad \bk(X,Y) := -\g(\bD_X \bN, Y),
\ee
for all vector fields $X, Y$ tangent to $\Hbf_t$.  Next, we define the {\sl deformation tensor} of the foliation by
\be
\label{deftens}
\bfpi := \L_{\bN} \, \g,
\ee
where $\L$ denotes the Lie derivative operator.
Finally, denoting by $\nablabf$ the
Levi-Civita connection of the slices $(\Hbf_t, \g_t)$, it is straightforward to show that
\be
\label{301}
\bfpi(\bN, \bN) = 0, \qquad \bfpi(\bN, X) = \nablabf_X \log n, \qquad \bfpi(X, Y)
= - 2 \, \bk(X, Y),
\ee
in which $X, Y$ are arbitrary vector fields orthogonal to the normal $\bN$.
Our main condition below will be stated in terms of $\bfpi$ and $n$, but it will
be convenient to also involve $\bk$ in our calculations.

We emphasize that all solutions to the field equations under consideration are assumed to be sufficiently
regular and have sufficient decay at spacelike infinity, so that all calculations below are justified.

An outline of this paper follows. Section~\ref{section2} is concerned with our main result about
the Bel--Robinson energy of  Weyl fields. In Section~\ref{section3}, we generalize this analysis
to Maxwell and Yang--Mills fields and discuss the case of scalar fields.
Finally, in Section~\ref{section4}, we discuss a class of spacetimes and state a conjecture for the Einstein equations.

%=================================================================

\section{Weyl fields on one-sided bounded spacetimes}
\label{section2}

\subsection{Properties of the Bel--Robinson tensor}

In all of this section (except for the final subsection), we assume that the spacetime $(\Mcal, \bg)$ is Ricci-flat, that is, in
components
\be
\label{ricciflat}
R_{\alpha\beta} = 0,
\ee
and we analyze the Bel--Robinson energy associated to the spacetime curvature $\bR$. At the end of this section,
we will explain how this analysis actually allows us to encompass Weyl fields on a given background.
In \eqref{ricciflat} and throughout, all Greek indices $\alpha, \beta,\ldots$ take values $0,\ldots, 3$ while 
Latin indices $i,j, \ldots$ take values $1,2,3$.

The following discussion depends upon the prescription of a future-oriented, time-like vector field
which, without restriction at this stage, is taken to coincide with the normal
$\bN$ to the foliation (introduced in Section~1).
In an orthonormal frame $\{\bfe_0,\bfe_1,\bfe_2,\bfe_3 \}$ with $\bfe_0=\bN$,
whose dual basis is denoted by $\{\bfe^0,\bfe^1,\bfe^2,\bfe^3\}$, the Lorentzian metric reads
$$
\g = \eta_{\a\b} \, \bfe^\a \otimes \bfe^\b := - \bfe^0 \otimes \bfe^0 + \bfe^1 \otimes \bfe^1
+ \bfe^2 \otimes \bfe^2 + \bfe^3 \otimes \bfe^3,
$$
while the volume form reads
$$
dV_\g = \bfe^0 \wedge \bfe^1 \wedge \bfe^2 \wedge \bfe^3 = \frac{1}{24} \,
\e_{\a\b\gamma\d} \, \bfe^\a \wedge \bfe^\b
\wedge \bfe^{\gamma} \wedge \bfe^\d.
$$
Here, as usual, we have set
$$
\e_{\a\b\gamma\d} = \left\{ \begin{array}{ll}
                        1, \quad & \a\b\gamma\d~\text{ even permutation of}~0123,
 \\
			   \hskip-.28cm  -1, & \a\b\gamma\d~\text{ odd permutation of}~0123,
\\
			    0, & \text{otherwise,}
                           \end{array} \right.
$$
and, for short, $\e_{ijk} := \e_{0ijk}$.

The left-- and right--hand {\sl Hodge duals} of the curvature tensor $\bR$, by definition, are
$$
{^\star R}_{\a\b\gamma\d} := \frac{1}{2} \e_{\a\b}{}^{\l\m} R_{\l\m\gamma\d},
\qquad
\quad R^\star{}_{\a\b\gamma\d} := \frac{1}{2} \e_{\gamma\d}{}^{\l\m} R_{\a\b\l\m},
$$
respectively. We recall that if the metric $\g$ is Einstein
(i.e.~the Ricci tensor $\bRic$ is proportional to the metric $\g$),
then ${^\star \bR}$ and ${\bR^\star}$ coincide.
In addition, when $\g$ is Ricci-flat as we assume in this section (see \eqref{ricciflat}),
then the tensors ${^\star\bR}$ and ${\bR^\star}$ possess the full
symmetries of the curvature tensor $\bR$.
(For details, see for instance~\cite[Sec.\ 4.6]{Penrose}.)

From the curvature tensor and its dual, we may then construct its {\sl electric and magnetic parts}
with respect to the vector field $\bN$, i.e.
$$
\bE(X,Y) := \g \big( \bR(X,\bN)\bN,Y \big),
\qquad \bH(X,Y) := \g({^\star\bR}\big( X,\bN)\bN,Y \big),
$$
respectively.
By the (skew) symmetries of the curvature tensor $\bR$ it follows that $\bE$ and
$\bH$ are symmetric and tangent to the foliation hypersurfaces $\Hbf_t$,
i.e.
$$
\bE(\bN, \cdot) = 0, \qquad \bH(\bN, \cdot) = 0.
$$
Moreover, since the metric is Ricci-flat, both $\bE$ and $\bH$ are
trace-free, in the sense that
$\delta^{ij} \bE(\bfe_i, \bfe_j) = \delta^{ij} \bH(\bfe_i, \bfe_j) = 0$.

To derive certain relations between the quantities $\bR$, ${^\star \bR}$, $\bE$, and $\bH$, 
it is convenient to define the {\sl reference Riemannian metric,}
\be
\label{999}
\gN := \g + 2 \, \g(\bN, \cdot) \otimes \g(\bN, \cdot),
\ee
associated with the Lorentzian metric and the given vector field.
This Riemannian metric determines the corresponding norms $\left| \, \cdot \, \right|_{\gN}$
and inner products $\langle \, \cdot \, , \, \cdot \, \rangle_{\gN}$
on tensor bundles over $\Mcal$.

\begin{proposition}[Decomposition of the Riemann tensor]
\label{REH}
Relative to an orthonormal frame $\{ \bfe_0, \bfe_1, \bfe_2, \bfe_3 \}$ with
$\bfe_0 = \bN$,
one has
\be
\label{eq1}
\aligned
& R_{ijk0} = - \e_{ijm} H_{km},\qquad \qquad \quad
 {^\star R}_{ijk0} = \e_{ijm} E_{km},
\\
& R_{ijkl} = - \e_{ijm} \e_{kln} E_{mn},
\qquad \quad
 {^\star R}_{ijkl} = -\e_{ijm} \e_{kln} H_{mn}.
\endaligned
\ee
Moreover, with respect to the Lorentzian metric $\g$ and the Riemannian metric $\gN$ one obtains
\begin{align}
\label{200}
|\bR|^2 &= 8 \, \big( |\bE|^2 - |\bH|^2\big),
 & |\bR|_{\gN}^2 = 8 \, \big( |\bE|^2 + |\bH|^2 \big),
\end{align}
respectively.
\end{proposition}

A proof of this result is provided shortly below.
To the Riemann curvature tensor, we can associate a notion of energy,
i.e.~the \text{\sl Bel--Robinson tensor}
$$
Q_{\a\b\gamma\d} := Q[\bR]_{\a\b\gamma\d} := R_{\a\l\gamma\m}
R_\b{}^\l{}_\d{}^\m +
{^\star R}_{\a\l\gamma\m} {^\star R}_\b{}^\l{}_\d{}^\m.
$$
It is straightforward to check that this tensor field is totally symmetric and trace-free.
Moreover, the components of $Q[\bR]$ are related to the electric
and magnetic parts of the curvature, as observed in \cite{AM,CK1,CK}.
These relations are conveniently expressed by setting
$$
A \cdot B = A_{\a\b} B^{\a\b},
\qquad
(A \wedge B)_\a = \e_\a{}^{\b\gamma} A_\b{}^\d B_{\d\gamma},
$$
and
$$
\aligned
(A \times B)_{\a\b}
= \, 
& \e_\a{}^{\gamma\d} \e_\b{}^{\m\nu} A_{\gamma\m} B_{\d\nu} + \frac{1}{3} (A
\cdot B)
g_{\a\b} - \frac{1}{3} (\tr A)(\tr B) g_{\a\b}
\\
= \, 
& A_\a{}^{\gamma} B_{\gamma\b} + A_\b{}^{\gamma} B_{\gamma\a} - \frac{2}{3} (A
\cdot B) \, g_{\a\b}
+ \frac{2}{3} (\tr A) (\tr B) \, g_{\a\b} - (\tr A)B_{\a\b} - (\tr B)A_{\a\b}.
\endaligned
$$
Observe that the $\times$ operation is symmetric.

\begin{proposition}[Decomposition of the Bel--Robinson tensor]
\label{REH-two}
Relative to an orthonormal frame $\{ \bfe_0, \bfe_1, \bfe_2, \bfe_3 \}$ with
$\bfe_0 = \bN$,
the relevant components of the Bel--Robinson tensor take the form
\be
\label{Q1}
\aligned
& Q_{0000} = |\bE|^2 + |\bH|^2, \qquad
Q_{i000} = 2 (\bE \wedge \bH)_i,
\\
& Q_{ij00} = - (\bE \times \bE)_{ij} - (\bH \times \bH)_{ij}
     + \frac{1}{3} \big( |\bE|^2 + |\bH|^2 \big) \, g_{ij}.
\endaligned
\ee
\end{proposition}

\begin{proof}[Proof of Propositions~\ref{REH} and \ref{REH-two}]
To derive~\eqref{eq1} we observe that, by definition,
$$
\aligned
{^\star R}_{ijk0}
&= \frac{1}{2} \e_{ij}{}^{\m\nu} R_{\m\nu k0}
\\
& = \e_{ij}{}^{0m} R_{0mk0} = (- \e_{ijm})
(- R_{m0k0}) = \e_{ijm} E_{mk}.
\endaligned
$$
The property $\star \star = - \id$ implies
$$
\aligned
R_{ijk0} & = - {^\star ({^\star R}_{ijk0})} = - {^\star (\e_{ijm} R_{m0k0})}
\\
& = - \e_{ijm} {^\star R_{m0k0}}
= - \e_{ijm} H_{mk}.
\endaligned
$$
The identities in the first line of~\eqref{eq1} then imply those in the second line, more precisely 
$$
\aligned
{^\star R_{ijkl}}
& = \frac{1}{2} \e_{ij}{}^{\m\nu} R_{\m \nu kl} = - \e_{ijm} R_{0mkl}
\\
& = \e_{ijm} R_{klm0}
= - \e_{ijm} \e_{kln} H_{mn} 
\endaligned
$$
and 
$$
\aligned
R_{ijkl} & = - {^\star ({^\star R}_{ijkl})} = - {^\star (- \e_{ijm} R_{0mkl})}
\\
& = - \e_{ijm} {^\star R_{klm0}}
= - \e_{ijm} \e_{kln} E_{mn}. 
\endaligned
$$ 

Due to the symmetries of the Riemann tensor, the above quantities suffice to
derive
$$
\aligned
|\bR|^2 & = R_{\a\b\gamma\d} R^{\a\b\gamma\d}
 = 4 R_{i0j0} R^{i0j0} + 4 R_{ijk0} R^{ijk0} + R_{ijkl} R^{ijkl}
\\
&= 4 E_{ij} E^{ij} - 4 \e_{ijm} \e_{ijn} H_{km} H^{kn} + \e_{ijm} \e_{kln}
\e_{ijo} \e_{klp} E_{mn} E^{op},
\endaligned
$$
thus
$$
\aligned
|\bR|^2
&= 4 E_{ij} E^{ij} - 8 H_{km} H^{km} + 4 E_{mn} E^{mn}
\\
&= 8 \, \big( |\bE|^2 - |\bH|^2 \big).
\endaligned
$$
On the other hand, to handle $|\bR|_{\gN}^2$, we observe that, in the calculation above,
we do not get a factor of $-1$ if we lower the $0$ index on the
second term, and this then changes the sign of the term
$|\bH|^2$. This completes the derivation of \eqref{200}.

Next, considering the Bel--Robinson tensor $Q_{\a\b\gamma\d}$, we immediately deduce from the definition that
$$
\aligned
Q_{0000} &= R_{0i0j} R_0{}^i{}_0{}^j + {^\star R}_{0i0j} {^\star
R}_0{}^i{}_0{}^j
\\
&= E_{ij} E^{ij} + H_{ij} H^{ij} = |\bE|^2 + |\bH|^2,
\endaligned
$$
while all the other terms can be computed in a similar fashion.
\end{proof}

%----------------------------------------------------------------------------------------

\subsection{Bel--Robinson energy inequality under a one-sided bound}

On a Ricci-flat spacetime, the Bel--Robinson tensor has the divergence-free property
\be
\label{dive}
\bD^\a Q_{\a\b\gamma\d} = 0,
\ee
as observed in Penrose and Rindler~\cite[Sec.\ 4.10]{Penrose}
and Christodoulou and Klainerman~\cite[Proposition~7.1.1]{CK}.
This property suggests that by integration one should be
able to control the Bel--Robinson energy
on an arbitrary slice of the foliation in terms of its values on some ``initial'' slice.
More precisely, we introduce the {\sl total Bel--Robinson energy} at time $t$ as
$$
\Q[\bR]_{\Hbf_t} := \int_{\Hbf_t} Q[\bR](\bN,\bN,\bN,\bN) \, dV_{\g_t},
$$
and we now derive a uniform estimate for $\Q_{\Hbf_t} $ which
solely involves a one-sided bound on the deformation tensor of the foliation.

First, we fix some notation and, by adopting an orthonormal frame with $\mathbf{e}_0 = \bN$, we express the deformation tensor defined in \eqref{nk}--\eqref{301}
in the form
\bel{picpts}
n \bfpi = \left( \begin{array}{cc}
                     0 & \nablabf  n \\
		     \nablabf n & - 2 n \bk
                    \end{array}
\right).
\ee
We also define the auxiliary tensor field
$$
\Lambda(n\bk) := (\tr (n\bk)) \, \id_3 - 2 n \bk,
$$
and the cubic polynomial of the real variable $\lambda$
\bel{poly}
P_{n \bfpi}(\lambda) := \det \big( \lambda \, \id_3 - \Lambda(n \bk) \big)
+ \big( \Lambda(n \bk) - \lambda \, \id_3 \big) \, \left( \nablabf n, \nablabf n \right),
\ee
referred to as the {\sl critical polynomial} of the spacetime foliation. This polynomial will arise naturally from the expression of the Bel--Robinson tensor; see \eqref{g-polynomial}, below. 

\begin{definition}
\label{rho}
The largest real root of the polynomial $P_{n \bfpi}$
is called the {\rm critical root} of $n \bfpi$ and denoted by $\rho(n \bfpi)$.
\end{definition}

\begin{theorem}[Bel--Robinson energy inequality under a one-sided bound]
\label{BR-estimate}
Given any vacuum Einstein spacetime endowed with a foliation $(\Hbf_t)_{t \in I}$ with lapse function $n$ and deformation tensor $\bfpi$,
one has
$$
\Q[\bR]_{\Hbf_{t_2}} \leq e^{3 K_{n \bfpi}(t_1,t_2)} \, \Q[\bR]_{\Hbf_{t_1}}, \qquad t_1 \leq t_2
$$
for all $t_1, t_2 \in I$, where
$$
K_{n \bfpi}(t_1,t_2): = \int_{t_1}^{t_2} \sup_{\Hbf_t} \rho(n \bfpi)  \, dt.
$$
\end{theorem}

\

For the proof of this result, we will derive (in Section~\ref{section23}, below) the algebraic
inequality
\be
\label{345}
- \frac{1}{2} \, Q_{\a\b00} \, n \pi^{\a\b} \le \rho(n \bfpi) \, Q_{0000}.
\ee
At this stage, we only check that~\eqref{345} implies the energy inequality stated in Theorem~\ref{BR-estimate}.
Namely, fixing any $t_1, t_2 \in I$ such that $ t_1 \leq t_2$ and applying Stokes theorem to the vector field
$\bP_\a = Q[\bR]_{\a\b\gamma\d} N^\b N^{\gamma} N^\d$ on the manifold with boundary
$\M_{[t_1, t_2]} := \bigcup_{t \in [t_1, t_2]} \Hbf_t$,
we obtain
$$
\aligned
\frac{3}{2} \int_{\M_{[t_1, t_2]}} \hskip-.3cm Q_{\a\b\gamma\d} \, \pi^{\a\b} N^{\gamma} N^\d \, dV_\g
&= \int_{\M_{[t_1, t_2]}} \hskip-.3cm \bD^\a \bP_\a \, dV_\g
\\
&= \Q_{\Hbf_{t_1}} - \Q_{\Hbf_{t_2}},
\endaligned
$$
which implies
$$
\aligned
\Q_{\Hbf_{t_2}}
& = \Q_{\Hbf_{t_1}} - \frac{3}{2} \int_{t_1}^{t_2} \int_{\Hbf_s} Q_{\a\b00} \, n \pi^{\a\b} \, dV_{\g_s} ds
\\
& \le \Q_{\Hbf_{t_1}} + 3 \int_{t_1}^{t_2} \int_{\Hbf_s} \rho(n \bfpi) \, Q_{0000} \, dV_{\g_s} ds
\\
& \le \Q_{\Hbf_{t_1}} + 3 \int_{t_1}^{t_2} \sup_{\Hbf_s} \left( \rho(n \bfpi) \right) \, \Q_{\Hbf_s} \, ds,
\endaligned
$$
where, in the first inequality, we have used~\eqref{345}.
Since sufficient regularity has been assumed on all solutions under consideration,
we deduce that the derivative of $\Q_{\Hbf_t}$ at $t = t_1$ satisfies
\[
\left. \frac{d}{dt} \Q_{\Hbf_t} \right|_{t=t_1}
= \lim_{t_2 \to t_1} \frac{\Q_{\Hbf_{t_2}} - \Q_{\Hbf_{t_1}}}{t_2 - t_1}
\le 3 \sup_{\Hbf_{t_1}} \left( \rho(n \bfpi) \right) \Q_{\Hbf_{t_1}}.
\]
Since this inequality holds for arbitrary $t_1 \in I$,
we may integrate it and arrive at the energy estimate of interest:
$$
\Q_{\Hbf_{t_2}} \leq \Q_{\Hbf_{t_1}} \exp \left(3 \int_{t_1}^{t_2} \sup_{\Hbf_t} \rho(n \bfpi) \, dt \right),
\qquad t_1 \leq t_2.
$$

%----------------------------------------------------------------------------------------

\subsection{Derivation of the optimal condition}
\label{section23}

This section is devoted to deriving the algebraic inequality~\eqref{345} and, more precisely,
\be
\label{345bis}
- \frac{1}{2} \, Q_{\a\b00} \, n\pi^{\a\b} \le C \, Q_{0000} 
\ee
for some constant $C>0$ depending on certain assumptions stated below.
By using the algebraic properties on $Q_{\a\b00}$ stated in
Proposition~\ref{REH-two}, above, we now derive conditions (on the eigenvalues of $n \bfpi$) implying~\eqref{345bis}.
We discuss the general case first, and then consider special cases of independent interest.

\begin{proposition}[Main algebraic condition on $\bfpi$]
\label{pi3}
The product $- \frac{1}{2} Q_{\a\b00} \, n\pi^{\a\b}$, with $n \bfpi$ decomposed as
in~\eqref{picpts}, can be controlled by the (double) eigenvalues $a_1 \leq a_2 \leq a_3$
of the $6 \times 6$--matrix $\Pi$, given by
$$
\Pi := n \begin{pmatrix} \Lambda(\bk) & - S(\bfpi) \\ S(\bfpi) &\Lambda(\bk) \end{pmatrix}, 
\qquad S_{ij}(\bfpi) := \e_{ijk} \pi^{k0}, 
$$
where the $3 \times 3$--matrices $\Lambda$ and $S$ are symmetric and
skew-symmetric, respectively.
More precisely, the estimate~\eqref{345bis} holds with $C=a_3$, and
the eigenvalues $a_1 \le a_2 \le a_3$ of $\Pi$ are nothing but the roots
of the polynomial $P_{n \bfpi}$ (defined in~\eqref{poly}), so that
$$
a_3 = \rho(n \bfpi).
$$
\end{proposition}

\begin{proof}
From Proposition~\ref{REH-two} we compute
\begin{align*}
-\frac{1}{2} Q_{\a\b00} \pi^{\a\b} &= Q_{ij00} k^{ij} - Q_{i000} \pi^{i0}
\\
&= \left( \delta_{ij} \left( |\bE|^2 + |\bH|^2 \right) - 2 \left( E_{ik} E_{jk}
+ H_{ik} H_{jk} \right)
\right) k^{ij} - 2 \epsilon_{ijk} E_{jl} H_{lk} \pi^{i0}
\\
&= \tr\bk \left( |\bE|^2 + |\bH|^2 \right) - 2 \tr \left( \bk \left( \bE^2 +
\bH^2 \right) \right)
- 2 S_{jk} E_{jl} H_{lk},
\end{align*}
so
\begin{align}
-\frac{1}{2} Q_{\a\b00} \pi^{\a\b}
&= \tr \left( \Lambda \left( \bE^2 + \bH^2 \right) - 2 S \bH \bE \right)
\nonumber
\\
&= \tr \left( \bE \Lambda \bE + \bH \Lambda \bH - \bE S \bH + \bH S \bE \right)
\label{Qpisum}
\end{align}
and, equivalently,
\be
\label{Qpi}
\aligned
- \frac{1}{2} \, Q_{\a\b00} \, n \pi^{\a\b}
& =
n \tr \left( \left( \bE, \bH \right)
\begin{pmatrix} \Lambda & - S \\ S & \Lambda \end{pmatrix}
\begin{pmatrix} \bE \\ \bH \end{pmatrix} \right)
\\
& =
\tr \left( \left( \bE, \bH \right)
\Pi
\begin{pmatrix} \bE \\ \bH \end{pmatrix} \right),
\endaligned
\ee
which provides us with a rather explicit expression for the left--hand side of \eqref{345bis}.

Observe that the $3 \times 3$ matrices $\Lambda$ and $S$ are symmetric and
skew-symmetric, respectively.
The linear map $\Pi$ given in the statement of the proposition
is a symmetric $6 \times 6$ matrix. Moreover, if we have an eigenvector of
$\Pi$ with eigenvalue $a$ and write it in the form $\begin{pmatrix} u \\ v \end{pmatrix}$ where $u, v
\in \RR^3$, then it is obvious that $\begin{pmatrix} v \\ -u \end{pmatrix}$ is a
distinct eigenvector of $\Pi$ associated with the same
eigenvalue $a$. Therefore, the matrix $\Pi$ has (at most) three distinct and real
eigenvalues, each of which appears
with multiplicity two. Denoting these eigenvalues by $a_1 \le a_2 \le a_3$,
from~\eqref{Qpi} and~\eqref{Q1}, we deduce:
\bel{a3eqn}
- \frac{1}{2} \, Q_{\a\b00} \, n\pi^{\a\b}
\le \max(a_1, a_2, a_3) \tr \left( \left( \bE \;\; \bH \right)
\begin{pmatrix} \bE \\ \bH \end{pmatrix} \right)
= a_3 \, Q_{0000}.
\ee

Next, we need to study the eigenvalues of the matrix $\Pi$ in terms of $\bfpi$.
First, we consider the case that the matrix $\Lambda$ is diagonal, so of the form $\diag[A, B, C]$, and we set
$$
a := - S_{23} = \nablabf_1 \log n, \qquad
b := - S_{31} = \nablabf_2 \log n, \qquad
c := - S_{12} = \nablabf_3 \log n.
$$
We therefore have
\be\label{Sigma}
\Pi =
n \begin{pmatrix}
A &0 &0 &0 &c &-b \\
0 &B &0 &-c &0 &a \\
0 &0 &C &b &-a &0 \\
0 &-c &b &A &0 &0 \\
c &0 &-a &0 &B &0 \\
-b &a &0 &0 &0 &C
\end{pmatrix},
\ee
and a simple calculation leads us to
\bel{detPi}
\det \left( \Pi - \lambda \, \id_6 \right) = P(\lambda)^2,
\ee
where
\be
\label{g-polynomial}
\aligned
P(\lambda)
:= & \lambda^3 - \lambda^2 n (A+B+C)
+ \lambda n^2 (AB + BC + CA - (a^2 + b^2 + c^2))
\\
& + n^3 \left( - A B C + Aa^2 + Bb^2 + Cc^2 \right)
\\
=
& \det \left( \lambda \id_3 - n \Lambda \right) +
\left( n A - \lambda \right) n^2 a^2 + \left( n B - \lambda \right) n^2 b^2
+ \left( n C - \lambda \right) n^2 c^2
\\
=
& \det \left( \lambda \id_3 - n \Lambda \right) +
\left( n \Lambda - \lambda \id_3 \right)
 \left( \boldsymbol{\nablabf} n, \boldsymbol{\nablabf}  n \right).
\endaligned
\ee
Hence, the eigenvalues of the matrix $\Pi$ are
given by the roots of the cubic polynomial, $P_{n \bfpi}$ defined in~\eqref{poly}.
From the fact that the latter expression in \eqref{g-polynomial} is $\mathrm{SO}(3)$--invariant, we deduce that
the same conclusion holds even when $\Lambda$ is not diagonal.
The discussion above implies that the polynomial  $P_{n \bfpi}$ has
three real roots $a_1 \le a_2 \le a_3$, and the result then follows from~\eqref{a3eqn}.
\end{proof}

\begin{remark} \rm
\label{rem:gpoly}
In~\eqref{g-polynomial}, only the squares of the coefficients $a, b, c$ appear.
Therefore, to bound the roots of the polynomial $P_{n \bfpi}$,
we will need to impose a bound on $|\nablabf  n|$. 
However, the fact that the polynomial $P_{n \bfpi}(\lambda)$ is invariant under the change of sign
of $a, b, c$ will later work at our advantage and
allow us to carry over our analysis to the case of Maxwell and
Yang--Mills fields with only minor changes.
\end{remark}

\subsection{Two special cases}

Although it is possible to write down the roots of a cubic polynomial
explicitly, in most cases it is more useful to estimate these roots by
simpler quantities. We first consider two illustrative special cases of vanishing $\bk$ and
vanishing $\nabla n$, respectively.

\begin{corollary}[The case of vanishing $\bk$]
\label{pi2}
Assume that the second fundamental form of the foliation slices vanish identically and, therefore,
$\bfpi$ is of the form
$$
\bfpi = \left( \begin{array}{cc}
                     0 & \nablabf \log n \\
		     \nablabf \log n & 0
                    \end{array}
\right).
$$
Then, the condition~\eqref{345bis} is satisfied with $C = | \nablabf n |$.
\end{corollary}

Under the assumptions in Corollary~\ref{pi2} we obtain a stronger conclusion and, actually, both an upper and a lower bound for the Bel--Robinson energy.
(Cf.~the two inequalities~\eqref{doublesided}, below.)

\begin{proof}
By setting $\Lambda := 0$ and $S_{jk} := \e_{ijk} \nablabf_i \log n$ in
Proposition~\ref{pi3}, we obtain
$$
- \frac{1}{2} \, Q_{\a\b00} \, n\pi^{\a\b} =
n \tr \left( \left( \bE, \bH \right)
\begin{pmatrix} 0 & -S \\ S & 0 \end{pmatrix}
\begin{pmatrix} \bE \\ \bH \end{pmatrix} \right)
\leq \max(a_1, a_2, a_3) \, Q_{0000}.
$$
The eigenvalues $a_1,a_2,a_3$ of the symmetric matrix
$$
\Pi = n \begin{pmatrix} 0 & -S \\ S & 0 \end{pmatrix}
$$
are the roots of the characteristic polynomial~\eqref{g-polynomial},
i.e.~$a_1 = - | \nablabf n |$, $a_2 = 0$, and $a_3 = | \nablabf n |$.
Therefore, we find
\bel{doublesided}
- | \nablabf n | \, Q_{0000}
\le -\frac{1}{2} Q_{\a\b00} \, n \pi^{\a\b}
\le | \nablabf n | \, Q_{0000},
\ee
and $Q_{\a\b00} \, n\pi^{\a\b}$ is controlled by $|\nablabf n|$ and $Q_{0000}$.
\end{proof}

From now on, we denote by
\be
\label{k123}
k_1 \leq k_2 \leq k_3
\ee
the eigenvalues of the second fundamental form $\bk$ of the foliation slices.

\begin{corollary}[The case of spatially constant lapse]
\label{pi1}
Assume that the lapse of the foliation is constant and, therefore, $\bfpi$ is of the form
$$
\bfpi = \left( \begin{array}{cc}
                     0 & 0 \\
		     0 & - 2 \bk
                    \end{array}
\right). 
$$
Then, if the inequality
\bel{k-condition}
n \left( k_2 + k_3 \right) \le C + n k_1
\ee
is satisfied for some real $C$, then the condition~\eqref{345bis} holds with this constant.
\end{corollary}

The condition~\eqref{k-condition} implies that
$$
n k_1 \le n k_2 \le n k_3 \le C,
$$
so that the principal curvatures of the hypersurface $\Hbf_t$  (weighted by the lapse $n$) are bounded
above. Geometrically speaking,
this condition means that the normal geodesics to the hypersurface $\Hbf_t$
are not focusing inwards too quickly.
It seems reasonable that such a condition is required if we wish for
our foliation not to collapse.

\begin{proof}
Since the second fundamental form is (real-valued and) symmetric, we may,
without loss of generality, assume that $\bk$ is diagonal of the form
$\mathrm{diag}[k_1, k_2, k_3]$ with, say, $k_1 \le k_2 \le k_3$.
We then deduce that
$\Lambda(\bk) = \mathrm{diag}[\lambda_1, \lambda_2, \lambda_3]$
with
\[
\lambda_1 := -k_1 + k_2 + k_3, \qquad
\lambda_2 := k_1 - k_2 + k_3, \qquad
\lambda_3 := k_1 + k_2 - k_3.
\]
Observe that $\lambda_3 \le \lambda_2 \le \lambda_1$ and, in view of
equation~\eqref{poly}, that the eigenvalues of the matrix $\Pi$ are $n \lambda_i$.
As such, the condition~\eqref{345bis} is satisfied if and only if
\[
n \lambda_1 \le C,
\]
which gives~\eqref{k-condition}.
\end{proof}

\begin{remark} \rm
Another way to view our result is to relax the ordering condition on the
eigenvalues and impose that
$K_i := n k_i - C \le 0$. Defining the vectors
$$
\mathbf{f}_1 := \begin{pmatrix} -1 \\ 1 \\ 1 \end{pmatrix}, \qquad
\mathbf{f}_2 := \begin{pmatrix} 1 \\ -1 \\ 1 \end{pmatrix}, \qquad
\mathbf{f}_3 := \begin{pmatrix} 1 \\ 1 \\ -1 \end{pmatrix},
$$
then the condition~\eqref{k-condition} may be rewritten in the form
$\langle \mathbf{K}, \mathbf{f}_i \rangle_{\RR^3} \le 0$. 
Therefore, the vector $\mathbf{K}$ lies in the intersection of three half-planes
through the origin in
$\mathbb{R}^3$. As such, $n \bk$ lies in a non-compact, triangular cone in $\RR^3$
with vertex at the point
$(C, C, C)$. Although the eigenvalues of $n \bk$ are bounded above, they may all
be arbitrarily large as
long as they remain within this cone. Observe that this cone contains the case
where $n \bk$ is a multiple ($\le C$) of the identity matrix
and, therefore, contains matrices that are not too far from a multiple of the
identity matrix.
The bound under consideration is very different from
the bounds considered in~\cite{KR, Wang}
which impose a restriction on $|\bk|$ and, hence,
impose both upper and lower bounds on the principal curvatures of the
hypersurfaces of the foliation.
\end{remark}

\begin{remark} \rm Under the conditions in Corollary~\ref{pi1}, assume also that the foliation consists of slices with constant mean curvature, that is, $\tr \bk = t$ on $\Hbf_t$ for all $t \in I$. If $\bk$ satisfies~\eqref{k-condition}, then
$2 n k_1(t) + C \ge n \tr\bk = n t$
and, therefore,
$$
n k_1(t) \ge \frac{1}{2} \left( n t - C \right).
$$
It follows that, as long as the lapse is bounded away from zero,
then $\bk$ is automatically also bounded from below as well as from above.
\end{remark}

%----------------------------------------------------------------------------------------

\subsection{Derivation of sufficient conditions}
\label{section24}

For general tensors $\bfpi$, we may combine the results of Corollary~\ref{pi2} and~\ref{pi1}
and arrive at a simplified condition on $\bfpi$, that is,
the condition~\eqref{pi12}, below,
which is in accordance with~\eqref{sigma-condition} derived
in the general case.

\begin{proposition}[A sufficient condition for general foliations. I]
A sufficient condition for~\eqref{345bis} to hold is that
\bel{pi12}
n \left( -k_1 + k_2 + k_3 \right) + | \nablabf n | \leq C,
\ee
which is automatically satisfied under the stronger
condition
\[
\max\big( n \left( -k_1 + k_2 + k_3 \right), | \nablabf n | \big) \leq \frac{C}{2}.
\]
\end{proposition}

\begin{proof}
From~\eqref{Qpisum}, we have
$$
- \frac{1}{2} n \, Q_{\a\b00} \pi^{\a\b}
= n \, \tr \left( \bE \Lambda \bE + \bH \Lambda \bH \right)
+ n \, \tr \left( \bH S \bE - \bE S \bH \right).
$$
We now estimate the two terms on the right-hand-side of this equation
separately. Corollary~\ref{pi1} implies that the first term is bounded above by
$n \left( -k_1 + k_2 + k_3 \right) Q_{0000}$, while Corollary~\ref{pi2} implies
that the second term is bounded above by $| \nablabf n | Q_{0000}$.
\end{proof}

Using estimates for eigenvalues of symmetric matrices, we deduce the following
assumptions that are sufficient to bound $Q_{\a\b00}\pi^{\a\b}$ without having
to compute all eigenvalues explicitly.

\begin{proposition}[A sufficient condition for general foliations. II]
\label{sigma-bounds}
By defining
\begin{align*}
\sigma_1 &:= n \left( -k_1 + k_2 + k_3 \right) + n |\pi^{20}| + n |\pi^{30}|, \\
\sigma_2 &:= n \left( k_1 - k_2 + k_3 \right) + n |\pi^{30}| + n |\pi^{10}|, \\
\sigma_3 &:= n \left( k_1 + k_2 - k_3 \right) + n |\pi^{10}| + n|\pi^{20}|,
\end{align*}
the condition~\eqref{345bis} holds with
\bel{sigma-condition}
C = \max(\sigma_1, \sigma_2, \sigma_3).
\ee
In particular, this holds if $n \bk$ satisfies~\eqref{k-condition}
and $|\nablabf n|$ is bounded.
\end{proposition}

\begin{proof}
Gershgorin's circle theorem allows us to bound the spectrum of the symmetric
matrix $\Pi$ and, precisely, states that the eigenvalues of $\Pi$
(with the notation~\eqref{Sigma}) are contained in the (closed) balls
$B_{n|b|+n|c|}(n A)$, $B_{n|a| + n|c|}(n B)$ and $B_{n |a| + n |b|}(n C)$.
Thus, an upper bound on the eigenvalues is given by
$\max(\sigma_1, \sigma_2, \sigma_3)$. Together with Proposition~\ref{pi3} this
leads to the desired estimate~\eqref{345bis}.
\end{proof}

%-----------------------------------------------------------------------------------------------

\subsection{Weyl fields on curved spacetimes}

Our previous results are now generalized to Weyl fields defined on a fixed Lorentzian manifold $(\Mcal, \g)$.
Recall from Christodoulou and Klainerman~\cite[Chapter~7]{CK} that a {\sl Weyl field} is a $(0, 4)$--tensor field, say $W_{\alpha\beta\gamma\delta}$, that has the same
symmetries as the Riemann tensor and, in addition, is trace-free, i.e.
$$
\aligned
W_{\a\b\gamma\d} & = W_{[\a\b][\gamma\d]}, \qquad W_{[\a\b\gamma]\delta} = 0,
\\
W^\a{}_{\b\a\delta} & = 0.
\endaligned
$$
Again, the left and right duals,
$$
{^\star W}_{\a\b\gamma\d} := \frac{1}{2} \e_{\a\b}{}^{\l\m} W_{\l\m\gamma\d},
\qquad
\quad W^\star{}_{\a\b\gamma\d} := \frac{1}{2} \e_{\gamma\d}{}^{\l\m} W_{\a\b\l\m},
$$
are equal, and are also Weyl fields. Finally, we define the {\sl Bel--Robinson tensor of the Weyl field}
to be
$$
Q := Q[\bW]_{\a\b\gamma\d} :=
W_{\a\l\gamma\m} W_\b{}^\l{}_\d{}^\m +
{^\star W}_{\a\l\gamma\m} {^\star W}_\b{}^\l{}_\d{}^\m.
$$

We assume that the spacetime $\Mcal$ is endowed with a foliation as stated in the introduction, 
and we define the {\sl total Bel--Robinson energy of the Weyl field} at time $t$
to be
$$
\Q[\bW]_{\Hbf_t} := \int_{\Hbf_t} Q[\bW](\bN,\bN,\bN,\bN) \, dV_{\g_t}.
$$
In the special case that the Bel--Robinson tensor associated with $\bW$ is divergence-free,
we immediately obtain an energy inequality for $\Q[\bW]$ along the same lines of the proof of
Theorem~\ref{BR-estimate},
specifically we obtain
\[
\Q[\bW]_{\Hbf_{t_2}} \leq \Q[\bW]_{\Hbf_{t_1}} e^{3 K_{n \bfpi}(t_1, t_2)} , \qquad t_1 \leq t_2.
\]
Observe that we no longer require that the background metric is Ricci-flat.
However, in general, $\operatorname{div} (Q[\bW]) \neq 0$ and, instead, we have the following conclusion.

\begin{theorem}[Energy inequality for Weyl fields under a one-sided bound]
\label{Weyl-estimate}
Let $(\Mcal, \g)$ be a spacetime endowed with a foliation $(\Hbf_t)_{t \in I}$ with lapse function $n$
and deformation tensor $\bfpi$ and, for all $t_1, t_2 \in I$ with $t_1 \leq t_2$, define
$$
K_{n \bfpi}(t_1,t_2) : = \int_{t_1}^{t_2} \sup_{\Hbf_t}  \rho(n \bfpi)  \, dt,
$$
where $\rho(n \bfpi)$ is the critical root of $P_{n \bfpi}$ (cf.~Definition~\ref{rho}).
Then, any Weyl field  $\bW$ defined on $\Mcal$ satisfies the energy inequality
\begin{align*}
\Q[\bW]_{\Hbf_{t_2}} \leq& \Q[\bW]_{\Hbf_{t_1}} e^{3 K_{n \bfpi}(t_1, t_2)}
 - \int_{t_1}^{t_2} e^{3 K_{n \bfpi}(t, t_2)}
\int_{\Hbf_t} n (\operatorname{div} Q[\bW])_{000} \, dV_{\bg_t}  dt, \qquad t_1 \leq t_2.
\end{align*}
\end{theorem}

\begin{proof}  Recalling that $\operatorname{div} (Q[\bW]) \neq 0$ in general,
we find the energy identity
\begin{align}
\Q[\bW]_{\Hbf_{t_2}} = \Q[\bW]_{\Hbf_{t_1}}
&- \frac{3}{2} \int_{t_1}^{t_2} \int_{\Hbf_t} n \, Q_{\a\b00} \pi^{\a\b} \, dV_{\bg_t} dt
\nonumber\\
&- \int_{t_1}^{t_2} \int_{\Hbf_t} n \, (\operatorname{div} Q[\bW])_{000} \, dV_{\bg_t} dt
\label{QW}
\end{align}
for all $t_1, t_2 \in I$ with $t_1 \le t_2$.
Arguing as in the passage after Theorem~\ref{BR-estimate}, we deduce from the inequality~\eqref{QW} that the derivative of $\Q$ at $t_1 \in I$ satisfies
\[
\left. \frac{d}{dt} \Q[\bW]_{\Hbf_t} \right|_{t=t_1}
\le 3 \sup_{\Hbf_{t_1}} \left( \rho(n \bfpi) \right) \Q[\bW]_{\Hbf_{t_1}}
- \int_{\Hbf_{t_1}} n \, (\operatorname{div} Q[\bW])_{000} \, dV_{\bg_{t_1}}
\]
and, therefore, for all $t_1 \leq t$
\[
\frac{d}{dt} \left( e^{- 3 K_{n \bfpi}(t_1, t)} \Q[\bW]_{\Hbf_t} \right)
\le - e^{- 3 K_{n \bfpi}(t_1, t)} \int_{\Hbf_t} n \, (\operatorname{div} Q[\bW])_{000} \, dV_{\bg_t}.
\]
Integrating in $t \in [t_1, t_2]$, we deduce that
\begin{align*}
e^{- 3 K_{n \bfpi}(t_1, t_2)} \Q[\bW]_{\Hbf_{t_2}} - \Q[\bW]_{\Hbf_{t_1}} &\le
- \int_{t_1}^{t_2} e^{- 3 K_{n \bfpi}(t_1, t)} \int_{\Hbf_t} n \, (\operatorname{div} Q[\bW])_{000} \, dV_{\bg_t} dt,
\end{align*}
as required.
\end{proof}

\begin{remark} \rm We cannot proceed any further without placing additional restrictions on the divergence of the Bel--Robinson tensor of the Weyl field, a quantity which, in the context of the Einstein equations, is determined by the matter content of the spacetime under consideration. Such terms with non-vanishing divergence do arise, even in the context of the {\sl vacuum} Einstein equations, as error terms generated by commuting Bianchi equations with vector fields (as developed in~\cite{CK}). 
\end{remark}

%======================================================================

\section{Maxwell and Yang--Mills fields on one-sided bounded spacetimes}
\label{section3}

\subsection{Maxwell equations}

We now consider a vacuum Maxwell field, represented by its curvature (or field strength)
$\bfF$, defined on a four-dimensional Lorentzian manifold $(\Mcal, \g)$ foliated as before.
By definition, $F$ satisfies the vacuum Maxwell equations
\be
\label{Maxeq}
d\bfF = 0, \qquad d \star \bfF = 0.
\ee
The energy-momentum tensor of the Maxwell field takes the form
$$
T_{\alpha\beta} = F_{\alpha\gamma} F_\beta{}^\gamma
                               - \frac{1}{4} \, g_{\alpha\beta} |\mathbf{F}|_{\g}^2.
$$
Letting $\mathbf{e}^{\alpha}$ be an orthonormal coframe, with dual frame
$\mathbf{e}_{\alpha}$, where we take $\mathbf{e}_0 = \bN$, then we divide the
components of $\mathbf{F}$ into electric and magnetic parts:
$$
E^i = F^{0i}, \qquad F^{ij} = \epsilon_{ijk} B_k.
$$
We define the {\sl total energy of the Maxwell field} at time $t$ to be
\[
\Q^{\text{MW}}[\bfF]_{\Hbf_{t}}
 := \int_{\Hbf_t} \T(\bN, \bN) dV_{\g_t} = \frac{1}{2} \int_{\Hbf_t} \left( |\bE|^2 + |\bB|^2 \right) dV_{\g_t}.
\]

A straightforward calculation in the above orthonormal frame then shows that
$$
\la \T, \bfpi \ra = \frac{1}{2} \tr\bfpi \left( |\bE|^2 + |\bB|^2 \right) -
\pi^{ij} \left( E_i E_j + B_i B_j \right) - 2 \pi^{0i} \left( \bE \times \bB
\right)_i
$$
Letting $S_{ij} := \epsilon_{ijk} \pi^{0k}$ and noting that $\pi_{ij} = - 2 k_{ij}$,
we deduce that
\begin{align*}
\la \T, \bfpi \ra
&= - \tr\bk \left( |\bE|^2 + |\bB|^2 \right) + 2 k^{ij} \left( E_i E_j + B_i B_j
\right) - 2 S_{jk} E_j B_k
\\
&= \tr \left[ \left( \bE^2 + \bB^2 \right) \left( 2 \bk - \tr\bk \, \id_3 \right) \right]
 - 2 \tr \left[ \bE S \bB \right].
\end{align*}
As before, letting $\Lambda(\bk) := (\tr\bk) \, \id_3 - 2 \bk$,
and noting that $S$ is skew-symmetric, we deduce that
$$
- \frac{1}{2} \la \T, \bfpi \ra
= \frac{1}{2} \tr \left( \bE \Lambda \bE + \bB \Lambda \bB + \bE S
\bB - \bB S \bE  \right).
$$
We therefore have
$$
- \frac{1}{2} \la \T, \bfpi \ra
= \tr \left(
\frac{1}{2} \begin{pmatrix} \bE, \bB \end{pmatrix}
\begin{pmatrix} \Lambda & S \\ - S & \Lambda \end{pmatrix}
\begin{pmatrix} \bE \\ \bB \end{pmatrix} \right),
$$
which leads to similar formulas to the ones found with the Bel--Robinson tensor.

We therefore define
$$
\widetilde{\Pi} := n \begin{pmatrix} \Lambda & S \\ - S & \Lambda \end{pmatrix}.
$$
Note, from the explicit form of the characteristic polynomial of $\Pi$ given in~\eqref{detPi}
and~\eqref{g-polynomial}, that changing the sign of the matrix $S$ does not
affect the eigenvalues of the matrix $\Pi$. It follows that the eigenvalues of
the matrix $\widetilde{\Pi}$ coincide with those of $\Pi$. We therefore deduce
that
$$
- \frac{1}{2} n \la \T, \bfpi \ra
\le \frac{1}{2} \max(a_i) \tr \left(
\begin{pmatrix} \bE, \bB \end{pmatrix}
\begin{pmatrix} \bE \\ \bB \end{pmatrix} \right)
= \max(a_i) \, T_{00}.
$$
Following the same proof as in our estimate for the Bel--Robinson energy,
we therefore deduce the following result.

\begin{theorem}[Energy inequality for Maxwell fields under a one-sided bound]
Let $(\Mcal, \g)$ be a spacetime endowed with a foliation $(\Hbf_t)_{t \in I}$ with lapse function $n$
and deformation tensor $\bfpi$ and, for all $t_1, t_2 \in I$ with $t_1 \leq t_2$, define
$$
K_{n \bfpi}(t_1,t_2) : = \int_{t_1}^{t_2} \sup_{\Hbf_t}  \rho(n \bfpi)  \, dt,
$$
where $\rho(n \bfpi)$ is the critical root of $P_{n \bfpi}$ (cf.~Definition~\ref{rho}).
Then, any vacuum Maxwell field  $\bW$ defined on $\Mcal$ satisfies the energy inequality
\[
\Q^{\mathrm{MW}}[\bfF]_{\Hbf_{t_2}} \le
\Q^{\mathrm{MW}}[\bfF]_{\Hbf_{t_1}} \, e^{K_{n \bfpi}(t_1,t_2)}, \qquad t_1 \leq t_2.
\]
\end{theorem}

%-------------------------------------------------------------------------------------------

\subsection{Yang--Mills equations}

As should be clear, the calculations for the Maxwell field may easily be
generalized to the case of Yang--Mills fields $F$ with compact gauge group, $G$.
In this case, the field strength is a two-form taking values in the associated Lie algebra
$\mathfrak{g}$. Since $G$ is compact, the Cartan--Killing form
$\boldsymbol{\kappa}$ is negative definite. Letting $T^i$ ($i = 1, \dots, \dim G$)
be an orthonormal basis for $\mathfrak{g}$ with respect to the Cartan--Killing
form, then the energy-momentum tensor of the Yang--Mills field takes the form
\be
\label{enYM}
T_{\alpha\beta} =
\sum_{i=1}^{\dim G}
\left( F_{\alpha\gamma}{}^i F_\beta{}^\gamma{}^i - \frac{1}{4} g_{\alpha\beta} F_{\gamma\delta}{}^i F^{\gamma\delta}{}^i \right).
\ee

The electric and magnetic parts of the curvature, along with the Yang--Mills energy are defined
in precise parallel with the case of Maxwell fields. As for Maxwell fields, we find that
$$
- \frac{1}{2} n \la \T, \bfpi \ra
= \frac{1}{2} n \sum_{i=1}^{\dim G} \tr \left(
\begin{pmatrix} \bE^i, \bB^i \end{pmatrix}
\begin{pmatrix} \Lambda & S \\ - S & \Lambda \end{pmatrix}
\begin{pmatrix} \bE^i \\ \bB^i \end{pmatrix} \right).
$$
From this result, again following the proof from the Bel--Robinson energy estimate and denoting
by $\Q^{\text{YM}}[\bfF]_{\Hbf_{t}} $ the {\sl total energy of the Yang-Mills field} $\bfF$ at time $t$, we deduce the following result.

\begin{theorem}[Energy inequality for Yang--Mills fields under a one-sided bound]
Let $(\Mcal, \g)$ be a spacetime endowed with a foliation $(\Hbf_t)_{t \in I}$ with lapse function $n$
and deformation tensor $\bfpi$ and, for all $t_1, t_2 \in I$ with $t_1 \leq t_2$, define
$$
K_{n \bfpi}(t_1,t_2) : = \int_{t_1}^{t_2} \sup_{\Hbf_t}  \rho(n \bfpi)  \, dt,
$$
where $\rho(n \bfpi)$ is the critical root of $P_{n \bfpi}$ (cf.~Definition~\ref{rho}).
Then, any Yang--Mills field (with compact gauge group) $\bfF$ defined on $\Mcal$ satisfies the energy inequality
\[
\Q^{\mathrm{YM}}[\bfF]_{\Hbf_{t_2}}
\le \Q^{\mathrm{YM}}[\bfF]_{\Hbf_{t_1}}
 \, e^{K_{n\bfpi}(t_1,t_2)}, \qquad t_1 \leq t_2.
\]
\end{theorem}

%-------------------------------------------------------------------------------------

\subsection{Scalar wave equation}

We now discuss the scalar wave equation, with the intent of also sharpening the standard energy inequality.
As we shall discover, however, this equation does not have a ``sufficiently rich'' algebraic structure for an analogue of Theorem~\ref{BR-estimate} to hold.

As before, $(\Mcal, \g)$ is a Lorentzian manifold, but we allow it to be of arbitrary dimension $d+1$,
and we assume a foliation by level sets $\Hbf_t$ of a time-function, as described in the introduction.
Let $\varphi \colon \Mcal \to \RR$
be a real-valued scalar field on $\Mcal$ that satisfies the wave equation
\be
\label{waveeq}
\Box_{\bg} \varphi = 0. 
\ee
Recall here that all solutions under consideration are smooth and have sufficient decay at spatial infinity. 
The {\sl energy-momentum tensor of the scalar field} $\varphi$ reads
$$
\T := d\varphi \otimes d\varphi - \frac{1}{2} \langle d\varphi, d\varphi \rangle \, \g,
$$
and the {\sl total energy of the scalar field} at time $t$ is
$$
\Q^{\mathrm{WE}}[\varphi]_{\Hbf_t}
 := \int_{\Hbf_t} \T( \bN, \bN ) \, dV_{\g_t}
= \frac{1}{2} \int_{\Hbf_t} |\bD \varphi|_{\gN}^2 \, dV_{\g_t}.
$$
We then have for $t_1 \leq t$
$$
\aligned
\Q^{\mathrm{WE}}[\varphi]_{\Hbf_{t}}
 -
\Q^{\mathrm{WE}}[\varphi]_{\Hbf_{t_1}}
 &=
- \int_{\M_{[t_1, t]}} D_\alpha \left( T^{\alpha\beta} N_\beta \right) dV_{\g}
= - \frac{1}{2} \int_{\M_{[t_1, t]}} \la \T, \bfpi \ra_{\g} \, dV_{\g},
\endaligned
$$
where $\la \T, \bfpi \ra_{\g}$ denotes the inner product (of the $(0, 2)$ tensor
fields) with respect to the Lorentzian metric $\g$.

Let us first recall the standard inequality.
We would normally estimate the second integral by
$$
\int_{\M_{[t_1, t]}} \langle \T, \bfpi \rangle \, dV_{\g} \le
\int_{\M_{[t_1, t]}} |\T|_{\gN} |\bfpi|_{\gN} \, dV_{\g}
$$
and, since
$|\T|_{\gN} \lesssim |\bD \varphi|_{\gN}^2$, deduce that
\begin{align*}
\left| \int_{\M_{[t_1, t]}} \langle \T, \bfpi \rangle \, dV_{\g} \right| &\lesssim
\int_{\M_{[t_1, t]}} |\bD \varphi(t)|_{\gN}^2 |\bfpi(t)|_{\gN} \, dV_{\g}
= \int_{t_1}^t \int_{\Hbf_s} |\bD \varphi|_{\gN}^2 |\bfpi|_{\gN}  n \, dV_{\g_s} ds
\\
&\le \int_{t_1}^t \sup_{\Hbf_s} \big( |n\bfpi| \big) \int_{\Hbf_s}
|\bD\varphi(s)|_{\gN}^2 \, ds
= 2 \int_{t_1}^t \sup_{\Hbf_s} \big( |n\bfpi| \big) \,  \Q^{\mathrm{WE}}[\varphi]_{\Hbf_s} \, ds.
\end{align*}
Gronwall's inequality leads us to the standard conclusion:
$$
\Q^{\mathrm{WE}}[\varphi]_{\Hbf_{t}} \lesssim \exp\Big( 2 \int_{t_1}^t \sup_{\Hbf_s} |n\bfpi|_\bg \, ds\Big)
\Q^{\mathrm{WE}}[\varphi]_{\Hbf_{t_0}}.
$$

Turning now to the derivation of a possibly weaker restriction on the deformation tensor,
we need to consider the quantity $\la \T, \bfpi \ra_{\g}$ more carefully. We have
$$
\aligned
\la \T, \bfpi \ra_{\g} &= T_{\alpha\beta} \pi^{\alpha\beta}
= \pi^{\alpha\beta} \left( D_\alpha \varphi D_\beta \varphi - \frac{1}{2} g_{\alpha\beta} |D
\varphi|_{\g}^2 \right)
\\
&= \left( \pi^{\a\b} - \frac{1}{2} \left( g^{\gamma\d} \pi_{\gamma\d} \right) g^{\a\b} \right)
D_{\a} \varphi D_{\b} \varphi
\\
&= \left( \pi_{\alpha\beta} - \frac{1}{2} \left( g^{\gamma\d} \pi_{\gamma\d} \right) g_{\alpha\beta} \right)
\left( g^{\a\epsilon} D_{\epsilon} \varphi \right) \left( g^{\b\eta} D_{\eta} \varphi \right)
\\
&= \left( \bfpi - \frac{1}{2} \left( \tr_{\g} \bfpi \right) \g \right)
\la \grad_{\g} \varphi, \grad_{\g} \varphi\ra,
\endaligned
$$
where we have defined the vector field $\grad_{\g} \varphi $ by, for arbitrary vector fields $\mathbf{X}$,
$$
\la \mathbf{X}, \grad_{\g} \varphi \ra_{\g} = \mathbf{X}(\varphi).
$$
Let $L$ be the operator defined by
\begin{align*}
\la X, L Y \ra_{\gN} &=
n \left( \frac{1}{2} \left( \tr_{\g} \bfpi \right) \g - \bfpi \right)(X, Y),
\end{align*}
or, equivalently, in local coordinates,
$$
L^\alpha{}_\beta :=
n \, g_{\bN}^{\alpha\gamma} \left( \frac{1}{2} \left( \tr_{\g}\bfpi \right) \, g_{\beta\gamma} - \pi_{\beta\gamma} \right).
$$
In terms of an orthonormal frame ${\bfe}_\alpha$ with $\bfe_0 := \bN$, $L$ has components
$$
\aligned
L^0{}_0 &= - \frac{n}{2} (\tr_{\g}\bfpi), \qquad
L^i{}_0 = - n \pi_{i0},
\\
L^0{}_i &= - n \pi_{i0},
\qquad \qquad
L^i{}_j = n \left( \frac{1}{2} (\tr_{\g}\bfpi) \delta_{ij} - \pi_{ij} \right)
\endaligned
$$
and, in particular, $\tr L = \frac{n}{2} (d-3) \tr_{\g}\bfpi$.

We denote the eigenvalues of $L$ by $\lambda_0, \dots, \lambda_d$, and then deduce that
$$
\aligned
- \frac{1}{2} n \la \T, \bfpi \ra_{\g}
& = \frac{1}{2} \la \grad_{\g} \varphi, L \grad_{\g} \varphi \ra_{\gN}
\\
& \le
\frac{1}{2} \max(\lambda_{\alpha}) |\grad_{\g}\varphi|^2_{\gN} = \max(\lambda_{\alpha}) \, T_{00}.
\endaligned
$$
Our energy inequality for the wave equation $\Box_{\g} \varphi = 0$ therefore becomes
\begin{align*}
\Q^{\mathrm{WE}}[\varphi]_{\Hbf_{t}} - \Q^{\mathrm{WE}}[\varphi]_{\Hbf_{t_1}}
&= - \frac{1}{2} \int_{\M_{[t_1, t]}} \la \T, \bfpi \ra_{\g} \, dV_{\g}
= - \frac{1}{2} \int_{t_1}^t \int_{\Hbf_s} n \la \T, \bfpi \ra_{\g} \, dV_{\g_s} \, ds
\\
&\le \int_{t_1}^t \left( \int_{\Hbf_s} \max(\lambda_{\alpha}) \, T_{00} \, dV_{\g_s} \right) \, ds
\\
&\le \int_{t_1}^t \max_{\Hbf_s} ( \max(\lambda_{\alpha}(s)))
\left( \int_{\Hbf_s} T_{00}(s) \, dV_{\g_s} \right) ds
\\
&= \int_{t_1}^t \left( \max_{\Hbf_s} \left( \max(\lambda_{\alpha}(s)) \right) \right) \Q^{\mathrm{WE}}[\varphi]_{\Hbf_{s}} \, ds.
\end{align*}

As such, if we impose that the eigenvalues of the operator $L$ are uniformly bounded above, in the sense that
\[
\la X, L X \ra_{\gN} \le C |X|_{\gN}^2,
\]
then we have the inequality
$$
\Q^{\mathrm{WE}}[\varphi]_{\Hbf_{t}} - \Q^{\mathrm{WE}}[\varphi]_{\Hbf_{t_1}}
\le \int_{t_1}^t C(s) \, \Q^{\mathrm{WE}}[\varphi]_{\Hbf_{s}} \, ds.
$$
Since $\Q^{\mathrm{WE}}[\varphi]_{\Hbf_{s}} \ge 0$, it follows that we have the energy inequality
$$
\Q^{\mathrm{WE}}[\varphi]_{\Hbf_{t}} \le
\Q^{\mathrm{WE}}[\varphi]_{\Hbf_{t_1}} \exp \left( \int_{t_1}^t C(r) \, dr \right).
$$

As the following example shows, however, the condition that the operator $L$
is bounded above implies that the second fundamental form is bounded both above and below.
Consider the case where the
lapse, $n$, is constant on the hypersurfaces $\Hbf_t$, i.e.~$n = n(t)$.
The only non-vanishing components of the deformation tensor are
$\pi_{ij} = - 2 k_{ij}$. We then have
\begin{align*}
- \frac{1}{2} n \la \T, \bfpi \ra &=
n \left( k^{ij} - \frac{1}{2} \tr\mathbf{k} \, \delta^{ij} \right) D_i \varphi D_j \varphi
+ \frac{n}{2} \tr\mathbf{k} \left( D_0 \varphi \right)^2
\end{align*}
Let $\Lambda^{ij} := k^{ij} - \frac{1}{2} \tr\mathbf{k} \, \delta^{ij}$.
Observe that $\Lambda$ is a symmetric, real $(d \times d)$ matrix, which may,
without loss of generality, be assumed diagonal of the form
$\Lambda = \mathrm{diag}[\lambda_1, \dots, \lambda_d]$.
Observe also that $\tr\bk = - \frac{2}{d-2} \tr\Lambda$, and we then have
\begin{align*}
- \frac{n}{2} \la \T, \bfpi \ra &=
n \tr \left( \Lambda D\varphi \otimes D\varphi \right)
- \frac{n}{d-2} \tr\Lambda \left( D_0 \varphi \right)^2
\\
&= \sum_{i=1}^d \lambda_i \left( D_i \varphi \right)^2
- \frac{n}{d-2} \tr\Lambda \left( D_0 \varphi \right)^2
\\
&= \left( D_0 \varphi, \dots, D_d \varphi \right) A \left( D_0 \varphi, \dots, D_d \varphi \right)^t,
\end{align*}
where
\[
A := \mathrm{diag} \left[ - \frac{\lambda_1 + \dots + \lambda_d}{d-2}, \lambda_1, \dots, \lambda_d \right]
\]
In order to bound $- \frac{1}{2} n \la \T, \bfpi \ra$ above in terms of $T_{00}$, we need the eigenvalues of the matrix $A$ to be bounded above.
This implies that we require the existence of a constant $C$ such that
\bel{scbd1}
\lambda_1, \dots, \lambda_d \le C,
\ee
but also
\bel{scbd2}
\lambda_1 + \dots + \lambda_d \ge - C(d-2).
\ee
However, the bounds~\eqref{scbd1} and~\eqref{scbd2} together imply that
\[
\lambda_1, \dots, \lambda_d \ge - C (2d - 3).
\]
It follows that, since $\Lambda$ is bounded both above and below, that the second fundamental form of the foliation, $\bk$, must also be bounded both above and below.

As such, our method does not generalize to the scalar field
when we only have a one-sided bound on the deformation tensor.
In particular, our methods depend strongly on
the algebraic structure of the energy momentum tensor.
The trace-free nature of the Bel--Robinson tensor and
the energy-momentum tensor for Maxwell and Yang--Mills fields
suggests that one might fare better with a conformally coupled scalar field,
which satisfies
$$
\Box_{\g} \varphi + C(d) R_{\g} \varphi = 0,
$$
where $C(d) = \frac{1-d}{4d}$ and $ R_{\g}$ denotes the scalar curvature. The energy momentum tensor of such a field takes the form
\begin{align*}
T_{\alpha\beta}
=
& \left( 1 + 2 C(d) \right) D_a \varphi D_b \varphi
- \frac{1}{2} g_{\alpha\beta} \left( 1 + 4 C(d) \right) |\bD\varphi|_{\g}^2
\\
& + \frac{1}{2} g_{\alpha\beta}C(d) R_{\g} \varphi^2 - C(d) R_{\alpha\beta} \varphi^2
 + 2 C(d) \varphi D_\alpha D_\beta \varphi
- 2 C(d) g_{\alpha\beta} \varphi \Box_{\g} \varphi.
\end{align*}
Unfortunately, it is now not possible to control $- n \la \T, \bfpi \ra$
due to the terms of the form $\varphi D_\alpha D_\beta \varphi$ in the energy-momentum tensor.

%==========================================================================

\section{Example and conjecture}
\label{section4}

\subsection{Example of spacetimes}

Our estimate in Theorem~\ref{BR-estimate} shows that, under certain
one-sided bounds, the Bel--Robinson energy of a vacuum solution of the Einstein equations cannot grow too quickly. If one were to assume a two-sided bound on the deformation tensor, then one would find, by a similar argument, a lower bound on the Bel--Robinson energy, showing that the energy, in addition, cannot decrease too quickly under a two-sided bound. In this section, we show that one can always choose a Lorentzian metric that satisfies a one-sided bound of our type, but for which the Bel--Robinson energy dies off as quickly as desired. More precisely, we have the following result.

\begin{theorem}
Given arbitrary $\eps>0$ and $\lambda \in (0,1)$, there exists a foliated, four-dimensional, Lorentzian manifold $(\Mcal, \g)$, with $\Mcal= \bigcup_{t \in I} \Hbf_t$, with the property that for $t_0 \in I$
\[
\frac{\Q_{\Hbf_{t_0 + \eps}}}{\Q_{\Hbf_{t_0}}} = \lambda.
\]
Hence,  
the Bel--Robinson energy dies off by an arbitrarily small factor (i.e.~$\lambda$)
in an arbitrarily short time (i.e.~$\eps$).
\end{theorem}

\begin{proof}
Our manifold is based on the Kasner metric. Recall that the Kasner metric takes the form
\[
\g = -dt^2 + t^{2\alpha} dx^2 + t^{2\beta} dy^2 + t^{2\gamma} dz^2,
\]
where we take $t > 0$ and $(x, y, z)$ to be coordinates on a three-dimensional flat torus, $T^3$, with the property that $\int_{T^3} \, dx \, dy \, dz = 1$. We choose to take $\alpha \le \beta \le \gamma$. The above metric is then Ricci flat if the conditions
\[
\alpha + \beta + \gamma = \alpha^2 + \beta^2 + \gamma^2 = 1
\]
are satisfied. Excluding the case where $\alpha = \beta = 0$, $\gamma = 1$ which gives a flat metric, then we necessarily have $\alpha < 0$ and $\gamma > \tfrac{1}{3}$. A straightforward calculation shows that the eigenvalues of the second fundamental form of the surface $\Sigma_t$ (of constant $t$)
 are $-\alpha/t, -\beta/t, -\gamma/t$. As such, the first root is unbounded above as $t \to 0$, while the last is unbounded below. It turns out that the behaviour of the root $-\gamma/t$ will be the most relevant to our discussion.

Noting that the lapse of this metric is equal to $1$,
the operator $\Lambda(\bk)$ defined earlier has eigenvalues
\[
\lambda_1 := \frac{2 \alpha - 1}{t}, \qquad
\lambda_2 := \frac{2 \beta - 1}{t}, \qquad
\lambda_3 := \frac{2 \gamma - 1}{t}.
\]
It follows that the critical root is
\[
\rho(n \bfpi) = \frac{2 \gamma - 1}{t}.
\]
A calculation of the curvature then shows that the magnetic part of the Weyl tensor of the Kasner metric vanishes, while the electric part is given by
\[
t^4 \, |\bE|^2 = \alpha^2 (\alpha-1)^2 + \beta^2 (\beta-1)^2 + \gamma^2 (\gamma-1)^2.
\]
It follows that the Bel--Robinson energy is
\[
t^3 \, \Q_{\Hbf_t} = \alpha^2 (\alpha-1)^2 + \beta^2 (\beta-1)^2 + \gamma^2 (\gamma-1)^2, \qquad t > 0.
\]

Our metrics are constructed as follows. Let $\eps, \lambda > 0$ be given, with $\lambda < 1$, and define the quantity
\[
a := - \frac{\eps \lambda^{1/3}}{1 - \lambda^{1/3}} < 0. 
\]
In terms of $a$ and for all $t > a$, we define the translated metric
\[
\g_a := -dt^2 + \left( t - a \right)^{2\alpha} dx^2 + \left( t - a \right)^{2\beta} dy^2 + \left( t - a \right)^{2\gamma} dz^2.
\]
Then, the Bel--Robinson energy of the set $t = \mathrm{constant} > a$ is given by
\[
\Q^a_{\Hbf_{t}} = \frac{\alpha^2 (\alpha-1)^2 + \beta^2 (\beta-1)^2 + \gamma^2 (\gamma-1)^2}{\left( t - a \right)^3}.
\]
In particular, letting $t_0 = 0$, we find 
\begin{align*}
\frac{\Q^a_{\Hbf_{t_0+\eps}}}{\Q^a_{\Hbf_{t_0}}} = \frac{\left( - a \right)^3}{\left( \eps - a \right)^3}
&= \left(\frac{\eps \lambda^{1/3}}{1 - \lambda^{1/3}}\right)^3 \left( \eps + \frac{\eps \lambda^{1/3}}{1 - \lambda^{1/3}} \right)^{-3}
\\
&= \left(\frac{\eps \lambda^{1/3}}{1 - \lambda^{1/3}}\right)^3 \left( \frac{\eps}{1 - \lambda^{1/3}} \right)^{-3}
= \lambda,
\end{align*}
and it follows that the metric $\g_a$, with $t_0 = 0$, has the required property.
\end{proof}

\begin{remark}
The metric $\g_a$ has constant lapse, so the critical root of the deformation tensor is given by
\[
\rho(n \bfpi) = \frac{2 \gamma - 1}{t - a}.
\]
For $t \ge t_0 \ge 0$, we therefore deduce that
\[
\int_{t_0}^t \rho(n \bfpi)(s) \, ds = \left( 2 \gamma - 1 \right) \log \left| \frac{t-a}{t_0 - a} \right|.
\]
This is bounded for any finite $t \ge t_0$, and therefore our inequality in Theorem~\ref{BR-estimate} applies
on any finite $t$-interval.
\end{remark}

%---------------------------------------------------------------------------------------------------

\subsection{Einstein equations}

This paper stems from recent work on the vacuum Einstein equations
by Klainerman, Rodnianski, and Wang~\cite{KR,Wang},
in which breakdown criteria were proposed for regular solutions to the Einstein
equations. Building upon the strategy already adopted by Grant and
LeFloch~\cite{GrantLeFloch1, GrantLeFloch2} to deal with the
injectivity radius of Lorentzian manifolds, we required in the present paper
solely one-sided geometric bounds on the spacetime. Our method should in principle be applicable
to the Einstein equations.

Recall that in harmonic coordinates $x^\m$, satisfying by
definition $g^{\a\b}\Gamma^\m_{\a\b} =0$ where $\Gamma^\m_{\a\b}$ denote the Christoffel symbols,
 the vacuum Einstein equations take the form of
a tensorial system of nonlinear wave equations
$$
0 = R_{\a\b} = - \frac{1}{2} \, g^{\gamma\d} \del_\gamma \del_\d g_{\a\b} +
F_{\a\b}(\g, \del \g).
$$
The metric $g_{\a\b} = m_{\a\b} + u_{\a\b}$ is typically expressed as
a small perturbation $\bu$ of the Minkowski metric $\bm$, and
the reduced vacuum Einstein equations above are rewritten as a nonlinear wave
equation for $\bu$, i.e.
$$
\Box_m \bu = N(\bu, \del\bu, \del^2\bu),
$$
where $\Box_\bm := m^{\a\b} \del_\a \del_\b$, and the nonlinearity $N$
consists of terms of the form $F(\bu) \cdot \bu \cdot \del^2 \bu$
and $F(\bu) \cdot \del \bu \cdot \del \bu$. Therefore,
relative to wave coordinates, breakdown of smooth solutions can not occur on some interval $[t_0, t_1)$
(see~\cite[Sec.~6.3]{Hormander} and~\cite[Sec.~3]{KN})
as long as
\[
\int_{t_0}^{t_1} \| \del \g (t) \|_{L^\infty(\Hbf_t)} \, dt < \infty.
\]

On the other hand, by taking advantage of the structure of the Einstein equations,
Klainerman and Rodnianski~\cite{KR} were able to prove that,
for CMC (constant mean curvature) foliations,
smooth solutions exist beyond time $t_1$ if
\[
\sup_{t \in [t_0,t_1)} \Big( \| \bk \|_{L^{\infty}(\Hbf_t)}
+ \| \nablabf \log n \|_{L^{\infty}(\Hbf_t)} \Big)
< \infty.
\]
This result was improved by Wang~\cite{Wang}
who assumed an integral rather than a $\sup$-norm condition:
\[
\int_{t_0}^{t_1} \Big( \| \bk \|_{L^\infty(\Hbf_t)} + \| \nablabf \log n
\|_{L^\infty(\Hbf_t)} \Big) \, dt < \infty.
\]
 
Based on the results established in the present paper, especially our new inequality for the Bel--Robinson energy,
we conjecture that the breakdown criterion for the Einstein equations can be drastically weakened
and a one-sided bound on $n\bfpi$ should be sufficient, that is, 
with the notation in Theorem~\ref{BR-estimate}:
\[
\int_{t_0}^{t_1} \sup_{\Hbf_t} \left( \rho(n \bfpi) \right) \, dt < \infty.
\]
In particular, this conjecture could be investigated in the case of CMC foliations. This objective seems to be realistic, based on the additional observation (made here by the authors) that the following a priori bounds hold (in terms of the Bel--Robinson energy, only):   
$$
\aligned
\| \bfpi \|_{L^6 (\Hbf_t)} + \| \bD \bfpi \|_{L^2 (\Hbf_t)} &\lesssim 1.
\\
\| n^{-1} \|_{L^\infty (\Hbf_t)} + 
\| \nablabf^2 n \|_{L^2 (\Hbf_t)} + \| \nablabf^2 (n^{-1}) \|_{L^2 (\Hbf_t)} &\lesssim 1,
\\
\| \dotn \|_{L^2 (\Hbf_t)} + \| \nabla \dot n \|_{L^2 (\Hbf_t)} &\lesssim 1.
\endaligned
$$
Finally, we refer to Grant and LeFloch \cite{GrantLeFloch1,GrantLeFloch2} for the derivation of injectivity radius estimates for Lorentzian manifolds that enjoy a one--sided bound, only.  

%==================================================================

\section*{Acknowledgments}

The first author (AYB) was supported by research stipend FS~506/2010 of the University of Vienna,
START-project Y237--N13 and FWF-project P20525 of the Austrian Science Fund.
The second author (JDEG) was supported by START-project Y237--N13
of the Austrian Science Fund and the University of Vienna.

The third author (PLF) was supported by the Centre National de la Recherche
Scientifique (CNRS) and the Agence Nationale de la Recherche through the grants ANR 2006-2--134423 (Mathematical Methods in General Relativity) and ANR SIMI-1-003-01 (Mathematical General Relativity. Analysis and geometry of spacetimes with low regularity).  
PLF is also grateful to the Mathematical Science Research Institute (MSRI, Berkeley), where this paper was completed
in April 2011.

%==================================================================


\begin{thebibliography}{999}

\bibitem{Andersson}
L. Andersson,
\emph{Bel--Robinson energy and constant mean curvature foliations},
Ann. Henri Poincar{\'e} \text{5} (2004), 235--244.

\bibitem{AM}
L. Andersson and V. Moncrief,
\emph{Future complete vacuum spacetimes}, in
``The Einstein equations and the large scale behavior of gravitational fields'', 
Birkh\"auser, Basel, 2004, pp.~299--330.

\bibitem{CK1}
D. Christodoulou and S. Klainerman,
\emph{Asymptotic properties of linear field equations in
Minkowski space},
Comm. Pure Appl. Math. \text{43} (1990), 137--199.

\bibitem{CK}
D. Christodoulou and S. Klainerman,
\emph{The Global Nonlinear Stability of the Minkowski Space},
Princeton Math. Series, Vol.~{41,}
Princeton Univ. Press, 1993.

\bibitem{GrantLeFloch1} J.D.E. Grant and P.G. LeFloch,
\emph{Null injectivity estimate under an upper curvature bound},
Preprint ArXiv:1008.5167.

\bibitem{GrantLeFloch2} J.D.E. Grant and P.G. LeFloch, in preparation.

\bibitem{Hormander}
L. H\"ormander,
\emph{Lectures on Nonlinear Hyperbolic Differential Equations},
Math{\'e}matiques \& Applications (Berlin)
[Mathematics \& Applications], Vol.~{26,} Springer Verlag, Berlin, 1997.

\bibitem{KN}
S. Klainerman and F. Nicol\'o,
\emph{The Evolution Problem in General Relativity},
Progress in Mathematical Physics, Vol.~{25,}
Birkh{\"a}user Boston, MA, 2003.

\bibitem{KR}
S. Klainerman and I. Rodnianski,
\emph{On the breakdown criterion in general relativity},
J. Amer. Math. Soc. \text{23} (2010), 345--382.

\bibitem{Penrose}
R. Penrose and W. Rindler,
\emph{Spinors and Space-Time. Vol.~1},
Cambridge Monographs on Mathematical Physics,
Cambridge Univ. Press, 1987.

\bibitem{Reiris}
M. Reiris,
\emph{The constant mean curvature Einstein flow and the Bel--Robinson energy},
Preprint arxiv:07053070v2.

\bibitem{Wang}
Q.~Wang,
\emph{Improved breakdown criterion for Einstein vacuum equations in CMC gauge},
Comm. Pure Appl. Math. 65 (2012), 21--76.   

\end{thebibliography}
\end{document}